\documentclass[aps,notitlepage,pra, nofootinbib,superscriptaddress]{revtex4-1}

\usepackage{mathpazo}

\usepackage{amsmath}
\usepackage{amsfonts,amssymb,amsthm, bbm,braket}
\usepackage{graphicx}   
\usepackage{subfigure}  
\usepackage{amsbsy} 
\usepackage[bold]{hhtensor} 
\usepackage{natbib}
\usepackage{algpseudocode}

\usepackage{multirow}

\newcommand\Tr{\mathrm{Tr}}

\newtheorem{corollary}{Corollary}
\newtheorem{theorem}{Theorem}
\newtheorem{lemma}{Lemma}
\newtheorem{proposition}{Proposition}

\usepackage[breaklinks=true]{hyperref}

\hypersetup{
  colorlinks   = true, 
  urlcolor     = blue, 
  linkcolor    = blue, 
  citecolor   = red 
}

\begin{document}


\title { Near-optimal quantum tomography: estimators and bounds}

\author{Richard Kueng}
\affiliation{Centre for Engineered Quantum Systems, School of Physics, The University of Sydney, Sydney, NSW, Australia}
\affiliation{Institute for Theoretical Physics, University of Cologne, Germany}
\affiliation{Institute for Physics, University of Freiburg, Germany}

\author{Christopher Ferrie}
\affiliation{
Center for Quantum Information and Control,
University of New Mexico,
Albuquerque, New Mexico, 87131-0001}
\affiliation{Centre for Engineered Quantum Systems, School of Physics, The University of Sydney, Sydney, NSW, Australia}

\date{\today}


\begin{abstract}
We give bounds on the average fidelity achievable by any quantum state estimator, which is arguably the most prominently used figure of merit in quantum state tomography.  Moreover, these bounds can be computed online---that is, while the experiment is running.  We show numerically that these bounds are quite tight for relevant distributions of density matrices.  We also show that the Bayesian mean estimator is ideal in the sense of performing close to the bound without requiring optimization.  Our results hold for all finite dimensional quantum systems.
\end{abstract}


\maketitle

\tableofcontents
\section{Introduction}

Inferring a quantum mechanical description of a physical system is equivalent to assigning it a quantum state---a process referred to as \emph{tomography}.  Tomography is now a routine task for designing, testing and tuning qubits in the quest of building quantum information processing devices \cite{Nielsen2010Quantum}.  In determining how ``good'' one is performing this task, a figure of merit must be reported.  By far the most commonly used figure of merit for quantum states is \emph{fidelity} \cite{Wootters1981Statistical,Jozsa1994Fidelity}.  Nowadays, fidelity is used to compare quantum states and processes in a wide variety of tasks, from quantum chaos to quantum control to the continuous monitoring of quantum systems \cite{Emerson2002Fidelity, Khaneja2005Optimal, Bagan2006Optimal, Emerson2007Symmetrized, Flammia2011Direct, DaSilva2011Practical, Cook2014Single}.  One might find it surprising, then, that the technique which optimizes performance with respect to fidelity is not known.

For $d$-dimensional state space,
\begin{equation}
\mathcal S :=  \left\{ \sigma \in L \left( \mathbb{C}^d \right):  \; \sigma \geq 0, \; \Tr (\sigma) = 1 \right\},
\end{equation}
the \emph{fidelity} between two states $\rho, \sigma \in \mathcal S$ is defined to be \cite{Wootters1981Statistical,Jozsa1994Fidelity},
\begin{equation}
F(\rho,\sigma) := \| \sqrt{ \rho} \sqrt{ \sigma} \|^2_1 = \left[\Tr  \sqrt{ \sqrt{\rho} \sigma \sqrt{\rho}} \right]^2.
\label{eq:fidelity}
\end{equation}
Define the \emph{average fidelity} with respect to some measure d$\rho$ as $\mathbb E_{\rho}[F(\rho,\sigma)]$\footnote{Expectation values will always be denote with a subscript which specifies the implicit distribution of variables being averaged over.}.
We want the average of this to be as large as possible.  Thus, the problem can be succinctly stated as follows:
\begin{equation}\label{the full problem}
\begin{aligned}
& {\text{maximize}}
& & \mathbb E_{\rho}[F(\rho,\sigma)] \\
& \text{subject to}
& & \Tr(\sigma) = 1, \\
& & & \sigma \geq 0.
\end{aligned}
\end{equation}

In the context of tomography, we think of $\rho$ as the ``true state'' and $\sigma$ as the estimated state.  An \emph{estimator} is a function from the space of data to quantum states $\sigma:\texttt{data} \mapsto \sigma(\texttt{data})\in\mathcal S$, where $\texttt{data}$ are the results of a sequence of quantum measurements.  Since both the true state and data are unknown, we take the expected value with respect to the joint distribution of $(\rho,\texttt{data})$ to obtain the average fidelity:
\begin{equation}\label{eq:jointaverage}
f(\sigma) = \mathbb E_{\rho,\texttt{data}}[F(\rho,\sigma(\texttt{data}))].
\end{equation}
We want this to be as large as possible.  The estimator which maximizes this quantity is equivalent to the estimator maximizing the following posterior average fidelity for every data set:
\begin{equation}
f(\sigma|\texttt{data}) = \mathbb E_{\rho|\texttt{data}}[F(\rho,\sigma(\texttt{data}))].
\end{equation}
An estimator which maximizes this is called a \emph{Bayes estimator}\footnote{The terminology and objective functions used here can be seen as standard generalizations of those familiar in decision theory.  See, e.g., \cite{Berger1985Statistical}.}.  Bayes estimators are useful both to understand Bayesian optimality and to provide upper bounds for the worst case performance.

Now here is the subtle and important point: the measurements performed, the data themselves and the distribution from which they were generated are not important once the posterior distribution has been calculated.  If we know the solution for every measure d$\rho$, then we know the solution for the posterior measure d$\rho|\texttt{data}$.  For brevity, then, we will drop this conditional information from now on and the problem reduces again to \eqref{the full problem}.

\section{Summary of Results} \label{sec:results}

In this work, we provide absolute benchmarks for the average fidelity performance of any tomographic estimation strategy by way of upper and lower bounds.    This is important because, in the field of quantum tomography, a common theme is to compare estimators.  Up to date many options are available:  linear inversion \cite{Nielsen2010Quantum}, maximum likelihood \cite{Hradil1997Quantum}, Bayesian mean \cite{BlumeKohout2010Optimal}, hedged maximum likelihood \cite{BlumeKohout2010Hedged}, and compressed sensing \cite{Gross2010Quantum, kueng_low_2014}---to name a few. Often estimators are compared by simulating measurements on ensembles of states drawn according to some measure and averaging the fidelity.  This can only provide conclusions about the relative performance of estimators.  Thus, our bounds can be used to benchmark the fidelity performance of other candidate estimators.  

We complement our theoretical findings with numerical experiments. These demonstrate the relative tightness of our bounds and, in particular, reveal that the Bayesian mean estimator is an excellent choice---owing to its near-optimal performance and ease of implementation.  Importantly, both the mean of the distribution and our bounds can be computed \emph{online}---that is, the estimator and its performance can be computed while data is being taken.   In the context of Bayesian quantum information theory \cite{BlumeKohout2010Optimal}, our findings lend credence to the standard approach of using the mean of the posterior distribution as an estimator is a near-optimal one.

We note that this problem has been solved for the case of a single qubit ($d=2$).  Bagan \emph{et al} \cite{Bagan2006Optimal} have given the optimal estimator (and measurement!) for any isotropic prior measure.  Unfortunately, by making heavy use of the Bloch representation of a qubit, the methods do not generalize. Whereas, our bound holds for all distributions of states in any dimension and coincides with the results of \cite{Bagan2006Optimal} for the case of a single qubit.

\subsection{Ensembles of pure states}

We first present the analytically soluble case of measures supported only on pure states.  Such a case is common in theoretical studies which average the performance of their protocols over the popular choice of the 
unique Haar invariant measure on pure states.  The solution is organized into the following theorem:
\begin{theorem}\label{thm:pure}
Choose an arbitrary dimension $d$ and assume that the integration measure $\mathrm{d} \rho$ is supported only on pure states. 
Then, 
the state which solves the optimization problem \eqref{the full problem} is the eigenvector of $\mathbb E_\rho[\rho]$ with maximal eigenvalue. It achieves 
a maximal fidelity of $ \left\| \mathbb E_\rho \left[ \rho \right] \right\|_\infty$.
\end{theorem}
The proof is a simple exercise in linear programming.  When $\rho$ is a pure state, the fidelity simplifies to $F(\rho,\sigma) = \Tr(\rho\sigma)$.  Linearity allows us to bring the expectation inside the trace so that the problem becomes
\begin{equation}\label{teh prob}
\begin{aligned}
& {\text{maximize}}
& & \Tr(\mathbb E_\rho[\rho] \sigma) \\
& \text{subject to}
& & \Tr(\sigma) = 1, \\
&&& \sigma \geq 0.
\end{aligned}
\end{equation}
The solution can be found in many textbooks covering linear programming---e.g. \cite{Boyd2004Convex}.  This solution also coincides with the one noted for a distribution supported on two states in \cite{BlumeKohout2006Accurate}.

\subsection{General measures on mixed states}

For measures with support on mixed states, the situation is markedly different.  Our main technical contribution are new upper bounds for this case.
We obtain them by replacing the fidelity function---which is notoriously difficult to grasp---in the main optimization problem \eqref{the full problem} 
by quantities that are easier to handle in full generality.  One rather straightforward approach to do so is to relate the fidelity function $f (\rho, \sigma)$ between arbitrary states $\rho,\sigma \in \mathcal{S}$ to corresponding Schatten-$p$-norm distances 
\begin{equation*}
\| \rho - \sigma \|_p = \left( \mathrm{Tr} \left( \left| \rho - \sigma \right|^p \right) \right)^{1/p},
\end{equation*}
with $ 1 \leq p \leq \infty$ and $| X| = \sqrt{ X^* X}$ for any $X \in L \left( \mathbb{C}^d \right)$. 
 This can be done by employing the well-known and often used Fuchs-van de Graaf inequalities \cite{Fuchs1999Cryptograhic}
\begin{align*}
1 - \sqrt{F \left( \rho, \sigma \right)} \leq \frac{1}{2} \| \rho - \sigma \|_1 
\leq \sqrt{1 - F (\rho, \sigma)} \quad \forall \rho, \sigma \in \mathcal{S}. 
\end{align*}
This inequality together with the hierarchy of Schatten-$p$-norms assures
\begin{equation}
F \left( \rho, \sigma \right) \leq 1 - \frac{1}{4} \left|\ \rho - \sigma \right\|_1^2 \leq 1 - \frac{1}{4} \left\| \rho - \sigma \right\|_2^2, \label{eq:fuchs_two_norm}
\end{equation}
for any two quantum states $\rho, \sigma \in \mathcal{S}$. 
Replacing the objective function in the central optimization problem \eqref{the full problem}
by such an upper bound results in a different optimization which admits a general analytic solution.
Clearly, such a relaxed optimum bounds the original figure of merit from above and allows us to establish our second main result.

\begin{theorem} \label{thm:boundfuchs}
For any finite dimension $d$ and any distribution $\mathrm{d} \rho$, the maximal average fidelity achieved by any estimator $\sigma \in \mathcal{S}$ obeys
\begin{equation}
\max_{\sigma \in \mathcal{S}}
\mathbb{E}_\rho
\left[
F ( \rho, \sigma )
\right]
\leq
1 - \frac{1}{4} 
\mathrm{Tr} \left( \mathbb{E}_\rho \left[ \rho^2 \right] - \mathbb{E}_\rho \left[ \rho \right]^2 \right). \label{eq:bound1}
\end{equation}
\end{theorem}
Note that the expression on the right hand side of \eqref{eq:bound1} can be interpreted as a non-commutative generalization of the variance of a
probability distribution.
Having already outlined the main ideas necessary to establish such a result, we refer to \autoref{sub:main_proof} for a complete proof.

Another way of establishing upper bounds on the average fidelity involves the concept of \emph{super-fidelity},
which provides the following upper bound on the fidelity \cite{Miszczak2009Sub}:
\begin{align}
F (\rho, \sigma) \leq \Tr \left( \rho \sigma \right) + \sqrt{1 - \Tr \left( \rho^2 \right)} \sqrt{1 - \Tr \left( \sigma^2 \right)}. 
\label{eq:super_fidelity}
\end{align}
Although more involved, we shall see that such an approach yields strictly better bounds than the ones presented in \autoref{thm:boundfuchs}.
For brevity, we define $\hat \rho:=\mathbb E_\rho[\rho]$ and $ p_\rho := \mathbb E_\rho\left[\sqrt{1-\Tr(\rho^2)}\right]$,
such that inequality \eqref{eq:super_fidelity} assures
\begin{equation}
\max_{\sigma \in \mathcal{S}} \mathbb{E}_\rho \left[ F \left( \rho, \sigma \right) \right]
\leq \max_{\sigma \in \mathcal{S}} 
\left( \mathrm{Tr} \left( \hat{\rho} \sigma \right) + p_{\rho} \sqrt{1- \mathrm{Tr} \left( \sigma^2 \right)} \right), \label{eq:super_fidelity_maximization}
\end{equation}
for any distribution $\mathrm{d} \rho$. 
Although more tractable than the original problem, the optimization on the right hand side still requires
solving a non-commutative maximization over all quantum states $\sigma \in \mathcal{S}$.
However, applying a corollary of the famous Birkhoff-von Neumann theorem---see e.g. \cite[Theorem 8.7.6]{Horn1990Matrix}---allows for restricting this optimization to density operators $\sigma$ that commute with the distribution's mean $\hat{\rho}$---see \autoref{lem:super_aux}  below.
If $\hat{r}_1,\ldots,\hat{r}_d$ denote the eigenvalues of $\hat{\rho}$ such a restriction assures that solving the right hand side of \eqref{eq:super_fidelity_maximization}
is equivalent to
\begin{align}
\textrm{maximize} & \quad \sum_{i=1}^d \hat{r}_i s_i + p_\rho \sqrt{1 - \sum_{i=1}^d s_i^2}
 \nonumber \\
\textrm{subject to} & \quad \sum_{i=1}^d s_i = 1, \label{eq:optimization1} \\
& \quad s_i \geq 0 \quad 1 \leq i \leq d, \nonumber
\end{align}
which is a commutative convex optimization problem. We refer to \autoref{lem:super_aux} below for a detailed proof of this assertion.  Note that, if the measure $\mathrm{d} \rho$ is supported exclusively on pure states, $p_\rho$ vanishes and \eqref{eq:optimization1} reduces to \autoref{thm:pure} which is tight.

In order to obtain analytical bounds for mixed states, we further relax \eqref{eq:optimization1} by replacing the non-negativity constraints ($s_i \geq 0$) by the weaker demand that the optimization vector $\left( s_1,\ldots,s_d \right)^T \in \mathbb R ^d$ is contained in the Euclidean unit ball---i.e. $\sum_{i=1}^d s_i^2 \leq 1$. 
As we shall show in \autoref{sec:geometry}, such a simplification is the tightest possible ellipsoidal relaxation of \eqref{eq:optimization1} and allows us to 
apply the method of Lagrangian multipliers in a straightforward fashion. 
Doing so results in the main theoretical statement of this paper.

\begin{theorem} \label{thm:bound}
For any finite dimension $d$ and any distribution $\mathrm{d} \rho$ over states, the fidelity achieved by any estimator $\sigma \in \mathcal{S}$ is bounded from above by
\begin{align}
\mathbb E_\rho[F(\rho,\sigma)] 
\leq
 \frac{1}{d} \left( 1 + \sqrt{d-1} \sqrt{d \left( \mathbb E_\rho \left[\sqrt{1- \Tr \left( \rho^2 \right)} \right]^2
+ \Tr \left( \mathbb E_\rho \left[ \rho \right]^2 \right) \right) -1} \right).
\label{eq:main_bound}
\end{align}
The matrix achieving this optimum corresponds to
\begin{equation}\label{eq:optimizer}
\sigma^\sharp
= \frac{1}{d} \mathbbm{1} 
+ \sqrt{\frac{d-1}{d \left( p_\rho^2 + \Tr \left( \hat{\rho}^2\right)\right) - 1}}
\left( \hat{\rho} - \frac{1}{d} \mathbbm{1} \right),
\end{equation}
where $\mathbbm{1} \in L \left( \mathbb{C}^d \right)$ denotes the identity matrix.
\end{theorem}

Again, we content ourselves here with outlining the proof architecture necessary to establish such a result and refer to \autoref{sec:proofs} for a detailed analysis.

Note that since we relaxed the maximization constraints, $\sigma^\sharp$ in general fails to be positive-semidefinite and is thus not a valid density operator, though we do not use it as such.  In particular, the bound is not tight when $\mathrm{d}\rho$ is supported only on pure states---as might be evident from the possibility of non-positive states arising from the $(\hat \rho-\frac{1}{d}\mathbbm 1)$ term in \eqref{eq:optimizer}.
On the other hand, the distribution is known and thus in the case of a distribution supported only on pure states, one should consult the exact solution in \autoref{thm:pure}.

Conversely, if $\sigma^\sharp$ happens to be a state, it also solves the optimization \eqref{eq:optimization1} and the analytical bound \eqref{eq:main_bound} exactly reproduces an a priori tighter one.  In all of our numerical experiments, some of which are presented below, this was indeed the case.

It is also worthwhile to point out that super-fidelity---the bound in \eqref{eq:super_fidelity}---and the actual fidelity coincide for one qubit, i.e. for $d = 2$ \cite{Miszczak2009Sub}. 
Also replacing positive semidefiniteness by bounded purity yields the same feasible set for that particular case.
Consequently the bound \eqref{eq:main_bound}
reproduces one of the main results in \cite{Bagan2006Optimal}:

\begin{corollary} \label{cor:bagan}
In the single-qubit case (i.e. $d = 2$) the bound \eqref{eq:main_bound}
exactly reproduces the maximum average fidelity 
 in \cite[Equation (2.9)]{Bagan2006Optimal} and $\sigma^\sharp$ is the optimal estimator.
\end{corollary}

Finally, we want to emphasize that establishing bounds on the average fidelity by using the super-fidelity instead of the Fuchs-van de Graaf inequalities 
leads to strictly better results:

\begin{corollary} \label{cor:majorization}
Let the dimension $d$ and the distribution $\mathrm{d}\rho$ over states be arbitrary.
Then, the bound presented in \autoref{thm:boundfuchs} {(Fuchs van-de Graaf inequality)}
is either trivial---i.e. equal to one---or it strictly majorizes the one presented \autoref{thm:bound} {(super-fidelity)}.
\end{corollary}

\section{Numerical Experiments}

Note that fidelity achieved by \emph{any} estimator is a lower bound on the one achieved by the optimal estimator.  A particularly convenient and generally well motivated \cite{BlumeKohout2006Accurate} estimator is the mean of the distribution $\hat \rho=\mathbb E_\rho[\rho]$.  Our findings underline that for distributions of states relevant to tomography, the mean is very near-optimal.  
In the context of tomography the mean is furthermore arguably the most convenient estimator,
since every other quantity of interest requires its calculation anyway.  

Finding an analytical expression for the posterior distribution is a very challenging
 problem, let alone performing the multidimensional integrals required for the calculation of the expectations above.   Thus, we turn to numerics.  In particular, we use the Sequential Monte Carlo (SMC) algorithm, which has been successfully applied to quantum statistical problems in the context of dynamical parameter estimation \cite{Chase2009Singleshot,Granade2012Robust,Wiebe2013Hamiltonian} and quantum state estimation \cite{Huszar2012Adaptive, Ferrie2014High,Ferrie2014Quantum}. Also, this algorithm is available as an open-source implementation in Python \cite{qinfer}.

Employing SMC allows us to perform the Bayesian updating and averaging. A complete and detailed discussion of the algorithm appears in Ref. \cite{Granade2012Robust} and thus we will not repeat the details here, but we will sketch the idea.  The algorithm starts  with a set of quantum states $\{\rho_j\}_{j=0}^{n}$, the elements of which are called \emph{particles}.  Here, $n = {|\{\rho_j\}|}$ is the number of particles and controls the accuracy of the approximation.  By approximating the prior distribution by a weighted sum of Dirac delta-functions,
\begin{equation}
\Pr(\rho) \approx \sum_{j=1}^{n} w_j  \delta(\rho - \rho_j),
\end{equation}
Bayes' rule then becomes
\begin{equation}
w_j \mapsto \Pr(\texttt{data}|\rho_j) w_j,
\end{equation}
followed by a normalization step.  The SMC algorithm is designed to approximate expectation values, such that
\begin{equation}
\mathbb E_\rho[f(\rho)] \approx \sum_{j=1}^n w_j f(\rho_j),
\end{equation}
for any function $f$.  In other words, the SMC algorithm allows us to efficiently compute the multidimensional integrals with respect to the measure defined by the posterior probability distribution.  We use this algorithm, as implemented by \cite{qinfer}, to numerically compute averages arising in simulated tomography experiments.  By doing so, we explore the efficacy of our claims for a variety of distributions relevant to practice and found natural in experimentation.

\begin{figure}
\includegraphics[width=0.6\columnwidth]{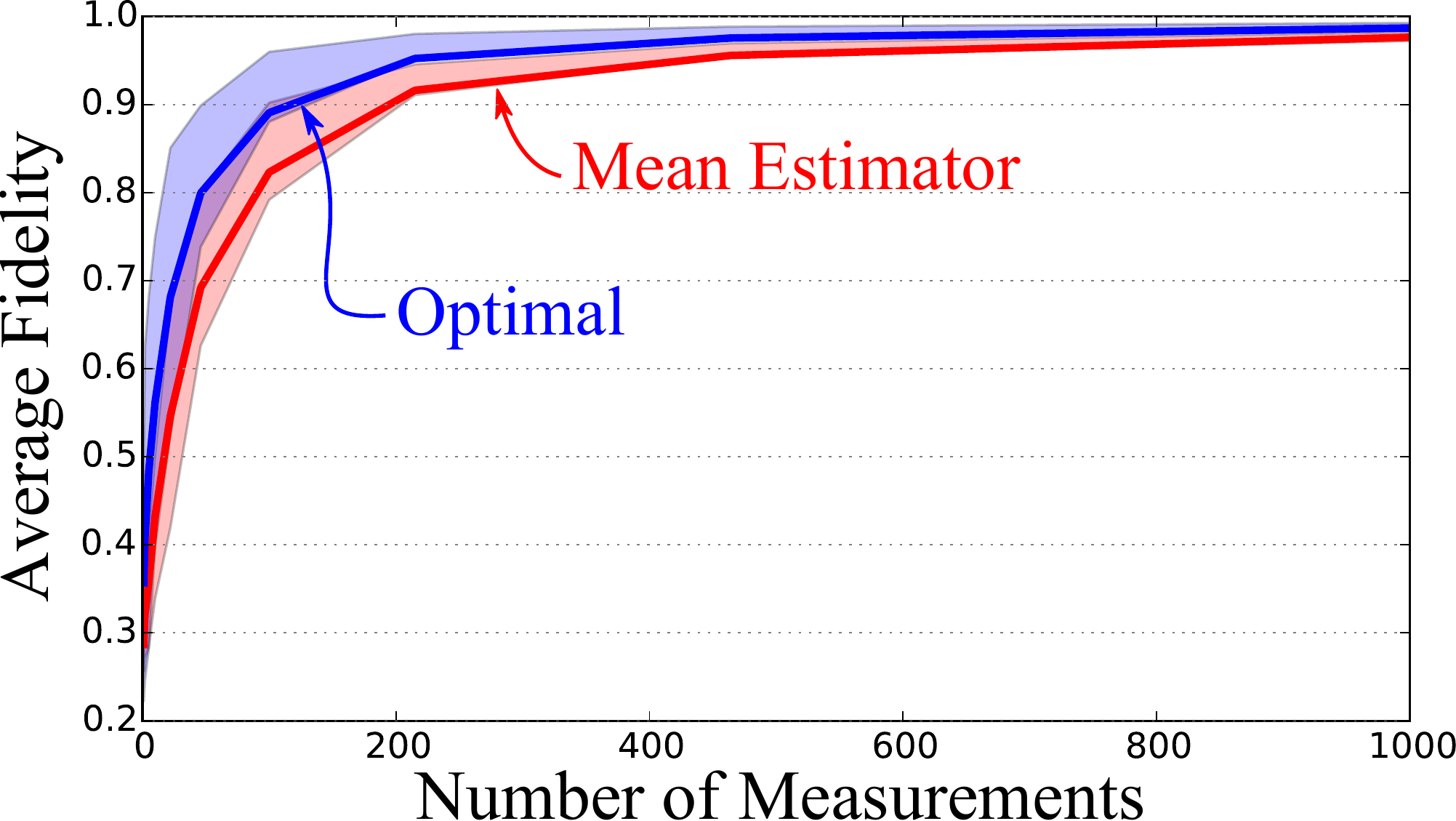}
\caption{\label{fig:pure}
The average fidelity as a function of the number of single-shot measurements of the Haar uniform measurement.  The prior distribution is here is also the Haar uniform measure on two qubits.  The lines are the medians and shaded areas the interquartile ranges over 100 trials.}
\end{figure}

Recall the sharp distinction between measures supported on pure states and those with full support.  We use the fact that \autoref{thm:pure} provides us with the optimal estimator in the former case to lend support to the claim that the mean estimator is a good candidate for a computationally simple, yet still near-optimal, alternative to solving the optimization problem in general.  In \autoref{fig:pure}, we present the results of numerical simulations on two qubits.  Plotted is the average fidelity achieved by the optimal estimator (see \autoref{thm:pure}) and the mean estimator $\mathbb E_\rho[\rho]$.  The average is taken with respect to a distribution that begins as the Haar invariant measure on pure states and is updated through simulated measurement data, where the measurement is the
``uniform POVM'' consisting of all pure states, distributed uniformly according to the Haar measure.  For independent measurements---i.e. local, non-adaptive ones---this measurement is optimal \cite[Theorem 3.1]{Holevo1982Probabilistic}.  We see that the mean estimator's fidelity tracks the optimal fidelity quite well.

\begin{figure}
\begin{tabular}{cc}
\begin{minipage}{8.8cm}
\includegraphics[width=\columnwidth]{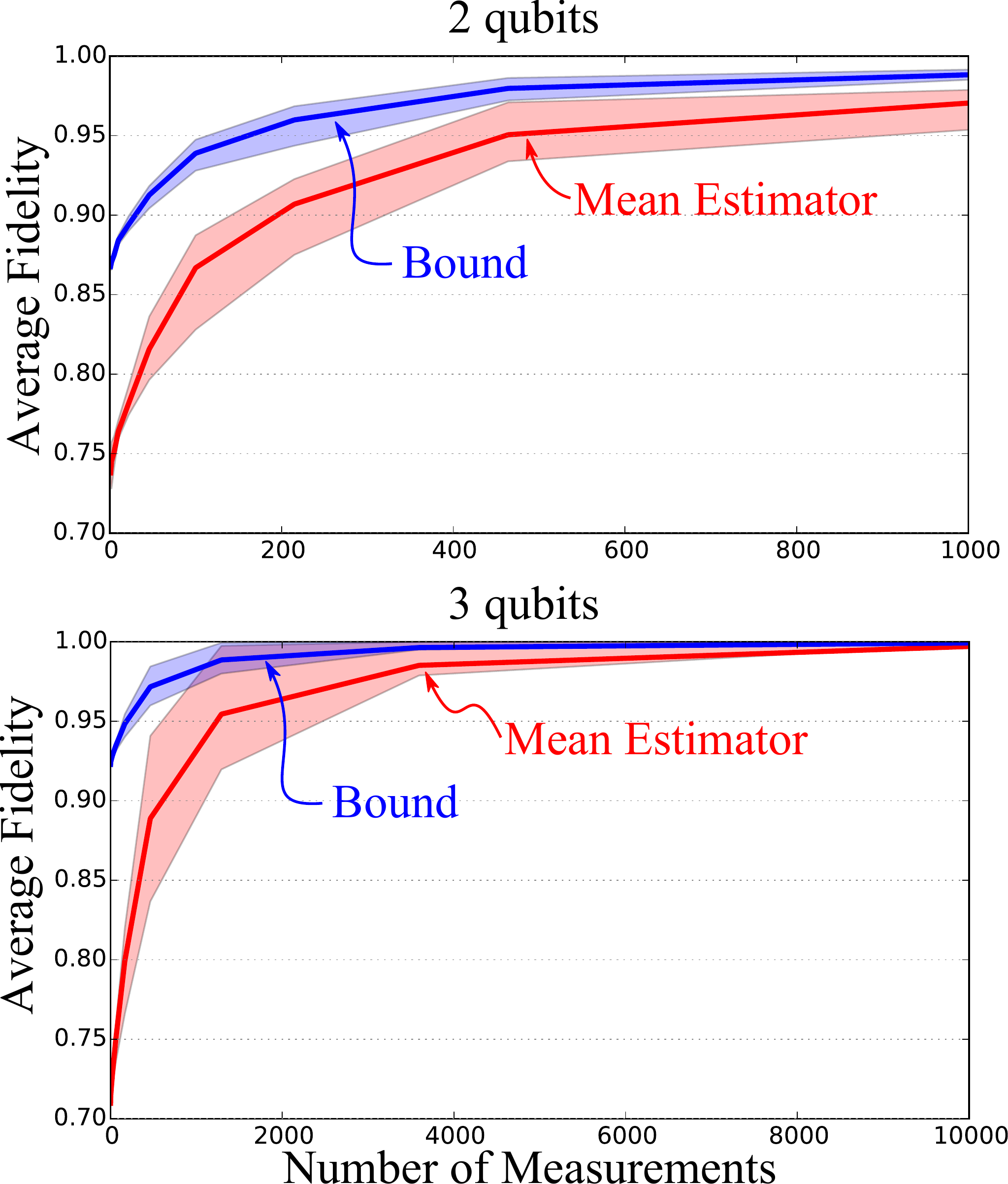}
\end{minipage}
&
\begin{minipage}{8.8cm}
\includegraphics[width=\columnwidth]{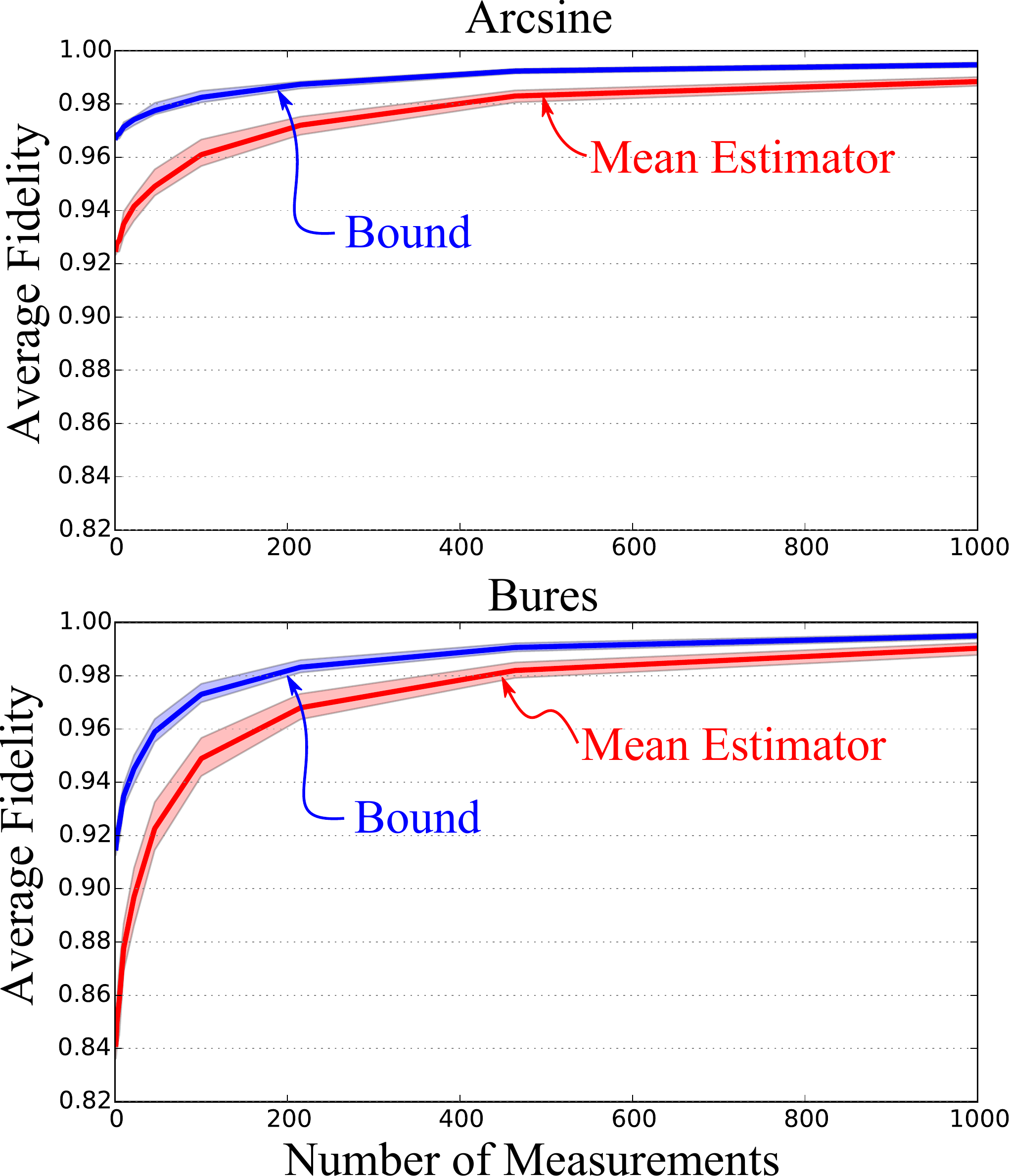}
\end{minipage}
\end{tabular}
\caption{\label{fig:mixed}
These plots depict the average fidelity as a function of the number of single-shot measurements of the Haar uniform measurements. \newline
\textbf{First column:}
The prior distribution is here is Hilbert-Schmidt measure on two and three qubit mixed quantum states.  \newline
\textbf{Second column:}  
The prior distribution for the upper plot is the \emph{Arcsine} distribution while for the lower plot the \emph{Bures} distribution was used---both are supported on two qubit mixed quantum states (again, see \cite{Zyckowski2011Generating} for a review of distributions of density matrices).  \newline
In all cases, the solid lines are the medians and shaded areas illustrate the interquartile ranges over 100 trials. }
\end{figure}

In \autoref{fig:mixed}, we plot the average fidelity of the mean estimator against our bound \eqref{eq:main_bound} for measures supported also on mixed quantum states.  Again, we simulate measurement data to get an accurate sense of how well the average fidelity of the mean estimator performs with respect to our bound for distributions relevant to tomography.  In this case, the \emph{prior} distribution is either the Hilbert-Schmidt measure (left column), or the  \emph{arcsine} and \emph{Bures} distributions \cite{Zyckowski2011Generating} for two qubits (right column).
In each case, many other natural distributions appear as we update our prior through Bayes' rule.  We see again that the mean estimator is a ``good'' estimator in that it comes close to the bound on the optimal fidelity and is the easiest  non-trivial average quantity to evaluate.

\section{Proofs} \label{sec:proofs}

In this section we provide detailed derivations and proofs of the statements presented in \autoref{sec:results}.

\subsection{A detailed proof of \autoref{thm:boundfuchs}}

Recall that in \autoref{thm:boundfuchs} we have claimed that the bound
\begin{equation}
\max_{\sigma \in \mathcal{S}} \mathbb{E}_\rho \left[ F (\rho, \sigma ) \right]
\leq 1 - \frac{1}{4} \mathrm{Tr} \left( \mathbb{E}_\rho \left[ \rho^2 \right] - \mathbb{E} \left[ \rho \right]^2 \right), \label{eq:fuchs_proof}
\end{equation}
is valid for any prior distribution $\mathrm{d} \rho$.
In order to derive such a statement, we start with inequality \eqref{eq:fuchs_two_norm}
\begin{equation*}
F ( \rho, \sigma ) \leq 1 - \frac{1}{4} \| \rho - \sigma \|_2^2,
\end{equation*}
which is a direct combination of the Fuchs-van de Graaf inequalities and the norm inequality $\| \cdot \|_2 \leq \| \cdot \|_1$. 
As such it is valid for any two states $\rho,\sigma \in \mathcal{S}$ which in turn assures that it remains valid upon taking expectations over $\mathrm{d} \rho$  on both sides:
\begin{equation}
\mathbb{E}_\rho \left[ F (\rho, \sigma ) \right]
\leq 1 - \frac{1}{4} \mathbb{E}_\rho \left[ \| \rho - \sigma \|_2^2 \right].
\end{equation}
Moreover,  we can optimize over $\sigma$ on both sides to obtain
\begin{align}
\max_{\sigma \in \mathcal{S}} \mathbb E_\rho \left[ F(\rho, \sigma) \right]
\leq 1 - \frac{1}{4} \min_{\sigma \in \mathcal{S}} \mathbb E_\rho \left[ \| \rho - \sigma \|_2^2 \right]. \label{eq:fuchs_proof_aux1}
\end{align}
The minimum on the right-hand side can in fact be calculated analytically. 
To this end, we define the function 
\begin{equation*}
f (\sigma) := \mathbb E_\rho \left[ \| \rho - \sigma \|_2^2 \right] 
= \Tr \left( \mathbb E_\rho \left[ \rho^2 \right] \right)
- 2 \Tr \left( \mathbb E_\rho [\rho] \sigma \right) + \Tr \left( \sigma^2 \right).
\end{equation*}
Note that $f (\sigma)$ is convex, because it corresponds to a weighted average of convex norm-functions $\| \sigma - \rho \|_2^2$ and its
 matrix-valued derivative corresponds to
\begin{align}
f' (\sigma) = - 2 \mathbb E_\rho \left[ \rho \right] + 2 \sigma. 
\end{align}
This derivative vanishes if and only if $\sigma^\sharp = \mathbb E_\rho \left[ \rho \right]$ 
holds and convexity of $f(\sigma)$ implies that this critical state corresponds to the unique minimum. The corresponding function value amounts to
\begin{align}
f \left( \sigma^\sharp \right)
= \Tr \left( \mathbb E_\rho \left[ \rho^2 \right] \right) 
- \Tr \left( \mathbb E_\rho \left[ \rho \right]^2 \right)
\end{align}
and reinserting this global minimum into \eqref{eq:fuchs_proof_aux1}
yields the desired bound \eqref{eq:fuchs_proof}. 

\subsection{A detailed derivation of \autoref{thm:bound}} 
\label{sub:main_proof}

Our main theoretical statement---\autoref{thm:bound}---follows from a three step procedure which was already briefly outlined in \autoref{sec:results}.

The first step invokes  the concept of super-fidelity \cite{Miszczak2009Sub} which assures
\begin{equation*}
\max_{\sigma \in \mathcal{S}} \mathbb{E}_\rho \left[ F \left( \rho, \sigma \right) \right] 
\leq \max_{\sigma \in \mathcal{S}} 
\left( \mathrm{Tr} \left( \hat{\rho} \sigma \right) + p_{\rho} \sqrt{1- \mathrm{Tr} \left( \sigma^2 \right)} \right) ,
\end{equation*}
with $\hat{\rho} = \mathbb{E}_\rho \left[ \rho \right]$ and $p_\rho = \mathbb{E}_\rho \left[ \sqrt{1- \mathrm{tr} \left( \rho^2 \right)} \right]$ for any distribution $\mathrm{d}\rho$.
As it turns out, the optimization on the right hand side of this equation is much more tractable than the original problem on the left hand side.
This is manifested by the following technical statement which is a direct consequence of the celebrated Birkhoff-von Neumann theorem.

\begin{lemma} \label{lem:super_aux}
Fix any  $p_\rho \geq 0$ and suppose that $\hat{\rho} \in \mathcal{S}$ is an arbitrary density operator with eigenvalue decomposition $\hat{\rho} = \sum_{i=1}^d \hat{r}_i |b_i \rangle \! \langle b_i|$. Then the optimization 
\begin{align}
\underset{\sigma \in L \left( \mathbb{C}^d \right)}{\textrm{maximize}} & \quad
 \mathrm{Tr} \left( \hat{\rho} \sigma \right) + p_{\rho} \sqrt{1- \mathrm{Tr} \left( \sigma^2 \right)}, \label{eq:super_aux1}\\
\textrm{subject to} & \quad \sigma \geq 0, \;\mathrm{Tr} (\sigma) = 1. \nonumber
\end{align}
is equivalent to solving
\begin{align}
\underset{s_1,\ldots,s_d \in \mathbb{R}}{\textrm{maximize}} & \quad \sum_{i=1}^d \hat{r}_i s_i + p_\rho \sqrt{1 - \sum_{i=1}^d s_i^2},
 \label{eq:super_aux2}\\
\textrm{subject to} & \quad \sum_{i=1}^d s_i = 1, \nonumber \\
& \quad s_i \geq 0 \quad 1 \leq i \leq d. \nonumber
\end{align}
Moreover, there is a one-to-one correspondence between any feasible array $(s_1,\ldots,s_d)$ of this problem and the density operator $\tilde{\sigma} = \sum_{i=1}^d s_i |b_i \rangle \! \langle b_i |$.
\end{lemma}

\begin{proof}
At the heart of this statement is an immediate corollary of the Birkhoff-von Neumann Theorem---see e.g. \cite[Theorem 8.7.6]{Horn1990Matrix}.
For $d \times d$ Hermitian matrices $\rho,\sigma$ this corollary assures
\begin{equation}
\mathrm{Tr} \left( \rho \sigma \right) \leq \sum_{i=1}^d r_i s_i, \label{eq:birkhoff}
\end{equation}
where $r_i$ and $s_i$ denote the eigenvalues of $\rho$ and $\sigma$, respectively, arranged in non-increasing order.
If $\hat{\rho}$ has eigenvalue decomposition $\hat{\rho} = \sum_{i=1}^d \hat{r}_i |b_i \rangle \! \langle b_i|$, the right hand side of \eqref{eq:birkhoff}
corresponds to $\mathrm{Tr} \left( \hat{\rho} \tilde{\sigma} \right)$ where $\tilde{\sigma} = \sum_{i=1}^d s_i |b_i \rangle \! \langle b_i|$.
Clearly, if $\sigma \in \mathcal{S}$ was a quantum state to begin with, so is $\tilde{\sigma}$, because the spectra of $\sigma$ and $\tilde{\sigma}$ coincide.
Moreover, such a definition assures that both states have equal purity, i.e. $\mathrm{Tr} (\sigma^2 ) = \mathrm{Tr} (\tilde{\sigma}^2)$. 
Consequently, for any feasible point $\sigma$ of the optimization \eqref{eq:super_aux1}, there is a $\tilde{\sigma}$ 
of the above form which admits a larger value in the optimization. Inserting the particular form of $\tilde{\sigma}$
into this program results in \eqref{eq:super_aux2}.
\end{proof}

In order to arrive at the bound presented in \autoref{thm:bound}, we employ one more relaxation which is going to allow us to solve the resulting problem analytically in full generality.
To be concrete, we replace the non-negativity constraints ($s_i \geq 0$) in \eqref{eq:super_aux2} by the weaker demand that the optimization vector $(s_1,\ldots,s_d)^T \in \mathbb{R}^d$ 
is contained in the Euclidean unit ball---i.e. $\sum_{i=1}^d s_i^2 \leq 1$.
Note that we explore the geometric properties of such a relaxation in \autoref{sec:geometry}. In a nutshell it corresponds to the tightest possible elliptical relaxation of the feasible set in \eqref{eq:super_aux1}.
By doing so, we arrive at the problem
\begin{align}
\underset{s_1,\ldots,s_d \in \mathbb{R}}{\textrm{maximize}} & \quad \sum_{i=1}^d \hat{r}_i s_i + p_\rho \sqrt{1 - \sum_{i=1}^d s_i^2}, \label{ALLCAPS!!!} \\
\textrm{subject to} & \quad \sum_{i=1}^d s_i = 1, \; \sum_{i=1}^d s_i^2 \leq 1,\nonumber 
\end{align}
which can be solved analytically via the method of Lagrangian multipliers:

\begin{lemma} \label{lem:super_aux2}
Let $\hat{r}_1,\ldots,\hat{r}_d$ denote the eigenvalues of any density operator and fix $p_\rho > 0$. Then the problem \eqref{ALLCAPS!!!} has a unique solution. The optimal value corresponds to
\begin{equation*}
 \frac{1}{d} \left( 1 + \sqrt{d-1} \sqrt{d \left(p_\rho^2
+ \Tr \left( \hat{\rho}^2 \right) \right) -1} \right)
\end{equation*}
and the array $(s_1^\sharp,\ldots,s_d^\sharp)$ achieving this optimum corresponds to the particular matrix
\begin{equation}
\sigma^\sharp
= \frac{1}{d} \mathbbm{1} 
+ \sqrt{\frac{d-1}{d \left( p_\rho^2 + \Tr \left( \hat{\rho}^2\right)\right) - 1}}
\left( \hat{\rho} - \frac{1}{d} \mathbbm{1} \right). \label{eq:optimal_solution}
\end{equation}
\end{lemma}

Note that this result together with the relaxations outlined in this section immediately implies \autoref{thm:bound}
upon inserting the definitions of $p_\rho$ and $\hat{\rho}$.
The assumption $p_\rho >0$ is furthermore non-critical, because, by definition, $p_\rho = 0$ if and only if $\mathrm{d} \rho$ is supported exclusively on pure states. 
This particular case, however, is already fully covered by \autoref{thm:pure}.

\begin{proof}[Proof of \autoref{lem:super_aux2}]
Throughout this proof we shall represent the eigenvalues of the density operator $\hat{\rho}$ as a vector
$\vec{\hat{r}} = \left( \hat{r}_1,\ldots,\hat{r}_d \right)^T \in \mathbb R ^d$. 
Likewise we shall encompass the scalar optimization variables $s_i$ in the vector $\vec{s} \in \mathbb R^d$. 
Furthermore, let $\vec{0} = (0,\ldots,0)^T$ and $\vec{1} = (1,\ldots,1)^T$
denote the ``all-zeros'' and ``all-ones'' vectors on $\mathbb{R}^d$, respectively. 
For $\vec{x},\vec{y} \in \mathbb R^d$, we will also make use of the standard inner product
$\langle \vec{x},\vec{y} \rangle = \sum_{i=1}^d x_i y_i$
and the vectorial inequality $\vec{x} \geq \vec{y}$ shall indicate component-wise inequality, i.e. $x_i \geq y_i$ for all $1 \leq i \leq d$. 

In such a vectorial form, the optimization problem \eqref{eq:super_aux2} corresponds to
\begin{align}
\textrm{maximize} & \quad f \left( \vec{s} \right) = \langle \vec{\hat{r}},\vec{s} \rangle + p_\rho \sqrt{1- \langle \vec{s},\vec{s} \rangle}, \nonumber \\
\textrm{subject to} & \quad g \left( \vec{s} \right) = \langle \vec{1},\vec{s} \rangle =0.
\label{eq:app_relaxed} \\
& \quad \langle \vec{s},\vec{s} \rangle \leq 1. \label{eq:app_sphere_constraint}\nonumber
\end{align}

Note that \eqref{eq:app_relaxed} is a convex optimization problem, as it requires maximizing a concave function over a convex set. 
As such, it has a unique maximum. 
One way of finding this maximum is to apply standard techniques such as the Karush-Kuhn-Tucker (KKT) multiplier method \cite{Boyd2004Convex} which are designed to take into account the inequality constraint  \eqref{eq:app_sphere_constraint}.

However, here we opt for a less direct but considerably more convenient and less cumbersome   approach:
we ignore the inequality constraint in \eqref{eq:app_relaxed} for now and 
employ the standard technique of Lagrangian multipliers (for equality constraints)
in order to find the unique critical point $\vec{s}^\sharp$ of the optimization. 
In a second step, we are going to verify that this vector strictly obeys the additional inequality constraint, we have ignored so far, i.e. $\langle \vec{s}^\sharp, \vec{s}^\sharp \rangle <1$. 
This in turn implies that said inequality constraint is not active at the critical point which in retrospect confirms that we were in fact right to ignore it in the first place.
Finally, the fact that we face a convex optimization problem assures that this unique critical point indeed yields the sought for global maximum of \eqref{eq:app_relaxed}. 

In order to find the critical point $\vec{s}^\sharp$ in question we define the Lagrangian function
\begin{align}
L (\vec{s} ) = f (\vec{s}) + \lambda g (\vec{s}),
\end{align}
where we have---as already announced---ignored the inequality constraint $\langle \vec{s},\vec{s} \rangle \leq 1$. 
As a consequence, $\lambda \in \mathbb{R}$ denotes the single Lagrangian multiplier associated with the remaining normalization constraint.
The necessary condition for an optimal solution of \eqref{eq:app_relaxed} then
reads
\begin{align}
\vec{\hat{r}}- \frac{ p_\rho \vec{s}}{\sqrt{1- \langle \vec{s}, \vec{s} \rangle }} + \lambda \vec{1} = \vec{0}.
\label{eq:app_lagrange}
\end{align}
Taking the inner product of this vector-identity with the ``all-ones'' vector $\vec{1}$ 
results in
\begin{equation}
0 
= \langle \vec{1},\vec{0} \rangle 
= \langle \vec{1}, \vec{\hat{r}} \rangle - \frac{p_\rho \langle \vec{1},\vec{s} \rangle}{\sqrt{1-\langle s,s \rangle}} + \lambda \langle \vec{1},\vec{1} \rangle  
= 1 - \frac{p_\rho}{\sqrt{1- \langle \vec{s},\vec{s} \rangle}} + d \lambda,
\label{eq:app_aux1}
\end{equation}
where we have used $\langle \vec{1},\vec{\hat{r}} \rangle = \sum_{i=1}^n \hat{r}_i = \Tr \left( \hat{\rho} \right) = 1$ and the normalization constraint, which likewise assures $\langle \vec{1},\vec{s} \rangle = 1$. 
This equation allows us to replace $\sqrt{1 - \langle \vec{s},\vec{s} \rangle}$ by 
$\frac{p_\rho}{1+d \lambda}$ and reinserting this into \eqref{eq:app_lagrange} results in the equivalent vector equation
\begin{align}
\vec{\hat{r}} - \left( 1 + d \lambda \right) \vec{s} + \lambda \vec{1} = \vec{0}.  
\end{align}
This can be readily inverted to yield
\begin{align}
\vec{s} = \frac{1}{1 + d \lambda} \left( \vec{\hat{r}} + \lambda \vec{1} \right).
\label{eq:app_lagrange_opt}
\end{align}
In order to determine the value of $\lambda$, we revisit \eqref{eq:app_aux1} 
which in combination with \eqref{eq:app_lagrange_opt} demands
\begin{align}
p_\rho^2 
&= (1 + d \lambda)^2 \left( 1 - \langle \vec{s},\vec{s} \rangle \right)
= (1 + d \lambda)^2 - \langle \vec{\hat{r}},\vec{\hat{r}} \rangle
- 2 \lambda \langle \vec{1},\vec{\hat{r}} \rangle - \lambda^2 \langle \vec{1},\vec{1} \rangle , \\
&= d (d-1) \lambda^2 + 2 (d-1) \lambda + 1 - \Tr \left( \hat{\rho}^2 \right),
\end{align}
where we have once more used $\langle \vec{1},\vec{\hat{r}} \rangle = 1$ as well as
$\langle \vec{\hat{r}}, \vec{\hat{r}} \rangle = \sum_{i=1}^n \hat{r}_i^2 = \Tr \left( \hat{\rho}^2 \right)$. 
This results in the quadratic equation
\begin{align}
\lambda^2 + \frac{2}{d} \lambda - \frac{1}{d(d-1)} \left( p_\rho^2 + \Tr \left( \hat{\rho}^2 \right) - 1 \right),
\end{align}
for $\lambda$ whose two possible solutions correspond to
\begin{align}
\lambda_\pm
&= - \frac{1}{d} \left( 1 \mp \sqrt{\frac{d \left( p_\rho^2 + \Tr \left( \hat{\rho}^2 \right) \right) - 1}{d-1}} \right).
\end{align}
Note that the argument of the square-root is non-negative, because the purity $\Tr \left( \hat{\rho}^2 \right)$ of any quantum state is lower-bounded by $1/d$. 
Also, the second solution $\lambda_-$ is vacuous, since it leads to an immediate contradiction. 
Indeed, it follows by inspection that $\lambda_- < - 1/d$ holds. Together with \eqref{eq:app_aux1}  this implies the contradictory relation
\begin{align}
\sqrt{1 - \langle \vec{s},\vec{s} \rangle} = \frac{ p_\rho}{1 + d \lambda_-}
< 0,
\end{align}
because $p_\rho$ is positive by assumption. 

Consequently we are left with one meaningful value $\lambda_+$ for the Lagrangian multiplier and 
inserting it into \eqref{eq:app_lagrange_opt}
yields the unique critical solution
\begin{align}
\vec{s}^\sharp
&= \frac{1}{d} \vec{1} + \sqrt{\frac{d-1}{d \left( p_\rho^2 + \Tr \left( \hat{\rho}^2 \right) \right) - 1}} \left( \vec{\hat{r}} - \frac{1}{d} \vec{1} \right).
\end{align}
Recall that throughout this proof we are exploiting a one-to-one correspondence between vectors $\vec{s} = (s_1,\ldots,s_n)^T \in \mathbb{R}^d$ and hermitian $d \times d$-matrices $\sigma = \sum_{i=1}^n s_i | b_i \rangle\!\langle b_i |$ that commute with $\hat{\rho}$. 
Consequently, the critical vector $\vec{s}^\sharp$ corresponds to the critical matrix
presented in \eqref{eq:optimal_solution}.

Plugging the critical point $\vec{s}^\sharp$ into the objective function $f(\vec{s})$ furthermore yields the corresponding critical function value
\begin{align}
f \left( \vec{s}^\sharp \right)
&= \langle \vec{\hat{r}}, \vec{s}^\sharp \rangle + p_\rho \sqrt{1 - \langle \vec{s}^\sharp, \vec{s}^\sharp \rangle}
= \frac{\langle \vec{\hat{r}}, \vec{\hat{r}} \rangle + \lambda_+ \langle \vec{1},\vec{\hat{r}} \rangle }{1 + d \lambda_+} 
+ \frac{p_\rho^2}{1+ d \lambda_+}
= \frac{d \left( p_\rho^2 + \Tr \left( \hat{\rho}^2 \right) \right) -1 + 1 + d \lambda_+}{d(1 + d \lambda_+ )}, \nonumber \\
&= \frac{1}{d} \left( 1 + \frac{ d \left( p_\rho^2 + \Tr \left( \hat{\rho}^2 \right) \right)-1}{1 + d \lambda_+} \right)
= \frac{1}{d} \left( 1 + \sqrt{d-1}\sqrt{d \left( p_\rho^2 + \Tr \left( \hat{\rho}^2 \right) \right) -1} \right), \nonumber 
\end{align}
where we have once more replaced $\sqrt{1 - \langle \vec{s}_+^\sharp,\vec{s}_+^\sharp \rangle}$ by $ \frac{p_\rho}{  (1 + d \lambda_+)}$ 
and combined that with the fact that $(1 + d\lambda_+) =  \sqrt{\frac{d \left( p_\rho^2 + \Tr \left( \hat{\rho}^2 \right) \right)-1}{d-1}}$ holds. 
 
With such a unique critical point $\vec{s}^\sharp$ at hand, we are now ready to show that it strictly obeys the inequality constraint $\langle \vec{s}^\sharp, \vec{s}^\sharp \rangle$ we have ignored so far. 
By employing the same equalities we have used in the previous paragraph, we can readily establish such a claim:
\begin{align}
\langle \vec{s}^\sharp, \vec{s}^\sharp \rangle
&= 1 - (1 - \langle \vec{s}^\sharp, \vec{s}^\sharp )
= 1 - \frac{p_\rho^2}{(1+ d \lambda_+)^2} 
< 1. \label{eq:app_strict_feasibility}
\end{align}
The strict inequality on the right follows from the fact that $p_\rho > 0$ holds by assumption.
This indeed establishes, that $\vec{s}^\sharp$ is also a critical point of the optimization problem \eqref{eq:app_relaxed}. Since this optimization corresponds to maximizing a concave function over a convex set, the unique critical point $\vec{s}^\sharp$ must correspond to the unique maximum of \eqref{eq:app_relaxed}.  
\end{proof}

\subsection{Detailed proofs of \autoref{cor:bagan} and \autoref{cor:majorization} }

We conclude the proof section with providing detailed proofs of the remaining statements, namely 
that \autoref{thm:bound} reproduces the main result in \cite{Bagan2006Optimal} for the particular case of a single qubit, i.e. $d=2$ (\autoref{cor:bagan})
and that the bounds presented in \autoref{thm:bound} are strictly better than the ones outlined in \autoref{thm:boundfuchs} (\autoref{cor:majorization}).

\begin{proof}[Proof of \autoref{cor:bagan}]

We start this section by pointing out that in the particular case of dimension $d=2$, the two relaxations we have employed
in the previous subsection are not relaxations at all. Indeed, for dimension two, fidelity and super-fidelity coincide, and moreover the sets $\left\{ \left(y_1,y_2 \right)^T \in \mathbb{R}^2: \; y_1 + y_2 = 1,\; y_1,y_2 \geq 0 \right\}$ and $\left\{ \left( y_1,y_2 \right) \in \mathbb{R}^2: \; y_1 + y_2 = 1,\; y_1^2 + y_2^2 \leq 1 \right\}$  coincide (this one-to-one correspondence is illustrated in \autoref{fig:geometry} below).
These low-dimensional equivalences assure that all the relaxations employed in the derivation of \autoref{thm:bound} are actually tight.
Consequently, in this particular low-dimensional case, we solve the actual problem of interest.

For deducing the claimed statement from this fact,
we consider Equation (2.9) in \cite{Bagan2006Optimal}:
\begin{equation}\label{bagan_F}
F = \frac12\left(1+\sum_\chi \|\mathbf{V}_\chi\|_2\right).
\end{equation}
Here $\chi$ simply means the the data generated via the measurement.
The vector $\mathbf{V}_\chi$ is defined as follows:
\begin{equation}
\mathbf{V}_\chi = \mathbb E_\rho[\underline{\mathbf{r}}\Pr(\chi|\rho)],
\end{equation}
where $\underline{\mathbf{r}}$ is related to the usual Bloch vector $\vec{r} = (x,y,z)$ via
\begin{equation}
\underline{\mathbf{r}} = \left(\sqrt{1-\|\vec{\hat{r}}\|_2^2},\vec{r}\right).
\end{equation} 
We point out that this $F$ is not the same average fidelity we have considered but the following quantity (which corresponds to our Eq. \eqref{eq:jointaverage} above):
\begin{equation}
F = \max_\sigma \mathbb E_\rho \left[\mathbb E_{\chi|\rho}[F(\rho,\sigma(\chi)]\right].
\end{equation}
Note however that, by employing Bayes' rule, this is equal to 
\begin{equation}
F = \max_\sigma \mathbb E_\chi \left[\mathbb E_{\rho|\chi}[F(\rho,\sigma(\chi)]\right],
\end{equation}
and thus maximizing the posterior average fidelity is equivalent to maximizing the total average fidelity.  Our bound applies directly to the former but trivially extends to the latter.

Thus, to establish Corollary 1, we need to extract the posterior average fidelity from the expressions above.  First, using Bayes' rule, we calculate
\begin{equation}
\mathbf{V}_\chi = \Pr(\chi)\mathbb E_{\rho|\chi}[\underline{\mathbf{r}}].
\end{equation}
Using the fact that $\|\vec{r}\|^2_2 = 2\Tr(\rho^2)-1$ and 
\begin{equation}
\Tr(\mathbb E_{\rho|\chi}[\rho]^2) = \frac12\left(1+\left\|\mathbb E_{\rho|\chi}\left[\vec{r}\right]\right\|_2^2\right),
\end{equation}
we find
\begin{equation}
\|\mathbf{V}_\chi\|^2_2 = \Pr(\chi)^2\left(2\mathbb E_{\rho|\chi}\left[\sqrt{1-\Tr(\rho^2)}\right]^2+2\Tr\left(\mathbb E_{\rho|\chi}[\rho]^2\right)-1\right).
\end{equation}
Plugging this back into \eqref{bagan_F}, we have
\begin{align}
F &= \frac12\left(1+\sum_\chi \Pr(\chi)\sqrt{2\mathbb E_{\rho|\chi}\left[\sqrt{1-\Tr(\rho^2)}\right]^2+2\Tr\left(\mathbb E_{\rho|\chi}[\rho]^2\right)-1}\right),\\
&= \frac12\left(1+\mathbb E_\chi\left[\sqrt{2\mathbb E_{\rho|\chi}\left[\sqrt{1-\Tr(\rho^2)}\right]^2+2\Tr\left(\mathbb E_{\rho|\chi}[\rho]^2\right)-1}\right]\right),\\
&= \mathbb E_\chi\left[\frac12\left(1+\sqrt{2\mathbb E_{\rho|\chi}\left[\sqrt{1-\Tr(\rho^2)}\right]^2+2\Tr\left(\mathbb E_{\rho|\chi}[\rho]^2\right)-1}\right)\right].
\end{align}
Thus, implied by the results of \cite{Bagan2006Optimal}, the maximum posterior average fidelity (dropping the $\chi$ for parallelism) is
\begin{equation}
\max_\sigma \mathbb E_\rho [F(\rho,\sigma)] = \frac12\left(1+\sqrt{2\left(\mathbb E_{\rho}\left[\sqrt{1-\Tr(\rho^2)}\right]^2+\Tr\left(\mathbb E_{\rho}[\rho]^2\right)\right)-1}\right).
\end{equation}
This coincides with our main result \eqref{eq:main_bound} for dimension $d=2$. 
\end{proof}

\begin{proof}[Proof of \autoref{cor:majorization}]
For notational simplicity, let us introduce the short-hand notation
\begin{align}
s_\rho := \Tr \left( \mathbb E_\rho \left[ \rho^2 \right] \right) - \Tr \left( \mathbb E_\rho \left[ \rho \right]^2 \right),
\end{align}
such that the bound presented in \autoref{thm:boundfuchs}
simply reads $\max_{\sigma \in \mathcal{S}} \mathbb E_\rho \left[ F (\rho, \sigma ) \right]
\leq 1 - \frac{s_\rho}{4}$. 
Note furthermore that $0 \leq s_\rho \leq 1$ holds. As already mentioned, the lower bound follows from invoking Jensen's inequality, while the upper bound is a simple consequence of the fact that the purity of any state is at most one. 
A vanishing $s_\rho$ would correspond to a trivial Fuchs-van de Graaf bound of one which is the first case instance covered by \autoref{cor:majorization}.
Therefore we can from now on safely assume that $s_\rho >0$ holds. 
Under this assumption we prove the second claim by starting with the bound presented in \autoref{thm:bound}
and upper-bounding it via a chain of inequalities which will ultimately lead to the bound presented in \autoref{thm:boundfuchs}. Indeed, pick any dimension $d$ and an arbitrary distribution $\mathrm{d} \rho$ over states. Then Jensen's inequality assures 
\begin{align}
\mathbb E_\rho \left[ \sqrt{1 - \Tr \left( \rho^2 \right)} \right]^2 \leq 1 - \mathbb E_\rho \left[ \Tr \left( \rho^2 \right) \right],
\end{align}
and the right hand side of expression \eqref{eq:main_bound} in \autoref{thm:bound} can be upper-bounded by
\begin{align}
\frac{1}{d} + \frac{\sqrt{d-1}}{d} \sqrt{ d -1 - d s_\rho},
\end{align}
because the square root function is monotonically-increasing on the positive reals.
Adding and subtracting $s_\rho$ in the last square root and once more invoking monotonicity
allows us to continue via
\begin{equation}
 \frac{1}{d} + \frac{\sqrt{d-1}}{d} \sqrt{(d-1)(1-s_\rho) - s_\rho}
 < \frac{1}{d} + \frac{d-1}{d} \sqrt{1-s_\rho}, \label{eq:fuchs_aux1}
\end{equation}
where we have used $s_\rho >0$ in the last line to obtain strict inequality. 
Since the square root is a concave function, the inequality
$\sqrt{1- s_\rho} \leq 1 - \frac{1}{2} s_\rho$ is valid for any $s_\rho \leq1$
and consequently
\begin{align}
\frac{1}{d} + \frac{d-1}{d} \sqrt{1-s_\rho}
\leq
 1 - \frac{d-1}{2d} s_\rho,
\end{align}
is true. Finally, we use the simple fact that $\frac{d-1}{d} \geq \frac{1}{2}$ holds for any $d \geq 2$ to arrive at 
$1 - \frac{1}{4} s_\rho$ which is just the Fuchs-van de Graaf bound.
Since a strict inequality sign connects the expressions in  \eqref{eq:fuchs_aux1},
 the claimed strict majorization follows. 
\end{proof}

\section{Geometric interpretation of the relaxation leading to \autoref{ALLCAPS!!!} \label{sec:geometry}}

Recall that in order to arrive at \autoref{thm:bound}, we have replaced the feasible set
\begin{align}
\Delta^{d-1} = \left\{ \vec{s} \in \mathbb{R}^d: \; \langle \vec{1},\vec{s} \rangle =1,\; \vec{s} \geq \vec{0}\right\},
\label{eq:app_state_space}
\end{align}
of the optimization problem \eqref{eq:optimization1} by 
\begin{align}
\mathcal{E}_{\Delta^{d-1}} = \left\{ \vec{s} \in \mathbb{R}^d: \; \langle \vec{1},\vec{s} \rangle = 1, \; \langle \vec{s},\vec{s} \rangle \leq 1 \right\},
\label{eq:app_relaxed_state_space}
\end{align}
which is a convex outer approximation of  $\Delta^{d-1}$.
This follows from the basic fact that $x^2 \leq x$ holds for any $x \in [0,1]$. 
Since the vector components $s_i$ of any $\vec{s} \in \Delta^{d-1}$ have to obey $s_i \in [0,1]$, we can readily conclude
\begin{align}
\langle \vec{s},\vec{s} \rangle
= \sum_{i=1}^d s_i^2 \leq \sum_{i=1}^d s_i = 1.
\end{align}
Note that the converse is true if and only if $d=1,2$---a fact which we have exploited in proving \autoref{cor:bagan}.

Geometrically, the former set corresponds to the standard simplex in $\mathbb{R}^d$. 
In this section we prove that the latter one 
is in fact the minimum volume covering ellipsoid of the standard simplex which furthermore corresponds to a $(d-1)$-dimensional Euclidean ball. 
For dimensions two and three this situation is illustrated in \autoref{fig:geometry}. 

\begin{figure}[h]
\begin{tabular}{cc}
\begin{minipage}{5cm}
\includegraphics[width=\columnwidth]{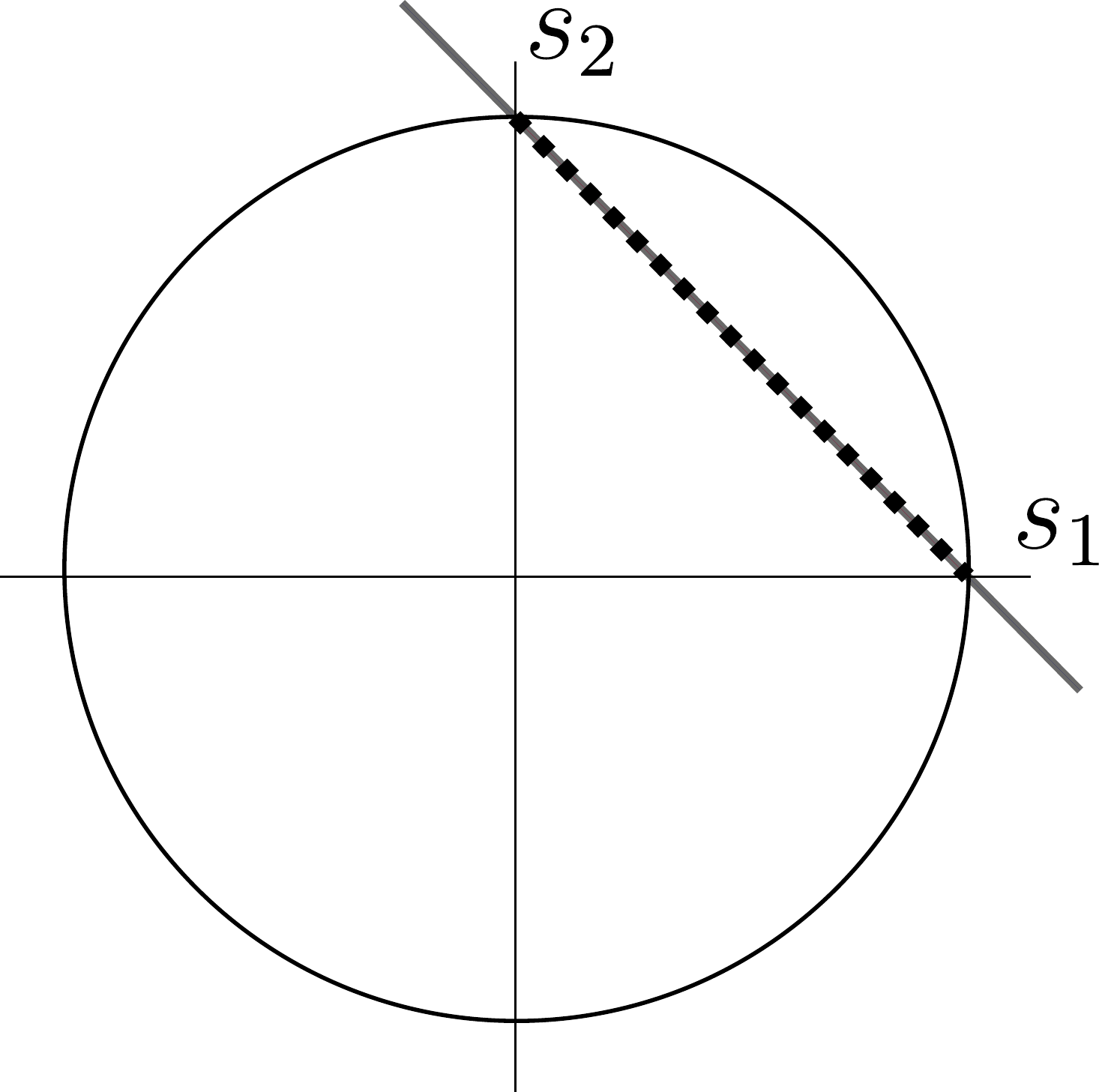}
\end{minipage}
\quad
\quad
\quad
&
\quad
\quad
\quad
\begin{minipage}{5cm}
\includegraphics[width=\columnwidth]{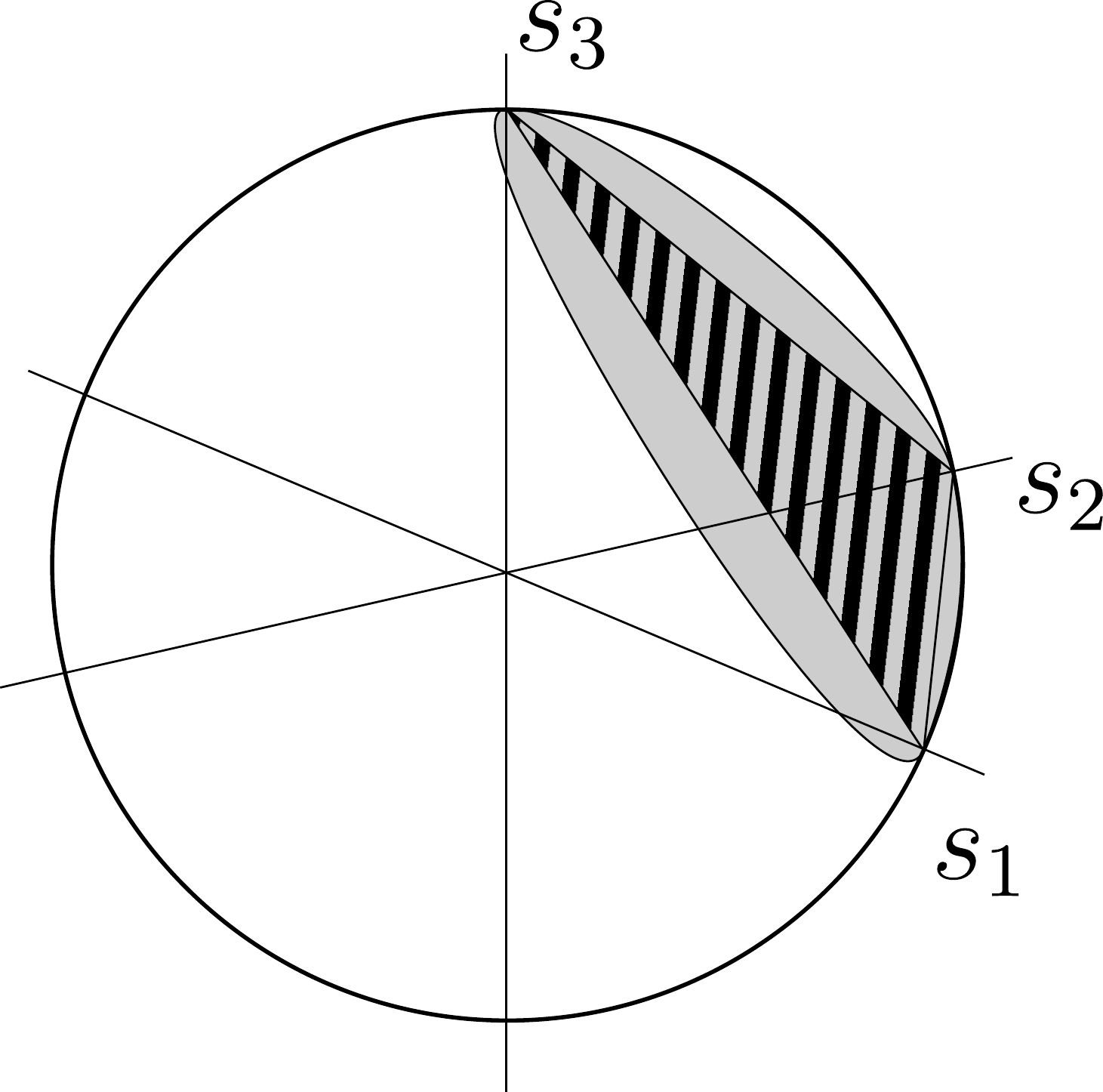}
\end{minipage}
\end{tabular}
\caption{\label{fig:geometry}
\emph{Geometric relation between the standard simplex $\Delta^{d-1}$ and its outer approximation $\mathcal{E}_{\Delta^{d-1}}$}:
Geometrically, the latter set corresponds to the minimum volume outer ellipsoid of the standard simplex.
The figure illustrates this relation for dimensions $d=2$ and $d=3$. Note that for $d=2$, the two sets coincide.}
\end{figure}

\begin{proposition}[Geometric nature of $\mathcal{E}_{\Delta^{d-1}}$] \label{prop:geometry1}
The convex outer-approximation $\mathcal{E}_{\Delta^{d-1}}$ of the $d$-simplex
corresponds to a $(d-1)$-dimensional Euclidean ball with radius $\sqrt{\frac{d-1}{d}}$
and center $\frac{1}{\sqrt{d}} \vec{1}$ which is contained in the $(d-1)$-dimensional hyperplane
$
 \mathcal{H}_{\vec{1},1}:= \left\{ \vec{s} \in \mathbb{R}^d: \; \langle \vec{1},\vec{s} \rangle =1 \right\}.
$
\end{proposition}

\begin{proof}
By definition, the set $\mathcal{E}_{\Delta^{d-1}}$ corresponds to the intersection 
of the Euclidean unit ball $\mathcal{B}_1 (0) = \left\{ \vec{s} \in \mathbb{R}^d:\; \langle \vec{s}, \vec{s} \rangle \leq 1 \right\}$ and the hyperplane $\mathcal{H}_{\vec{1},1}$.
This assures $\mathcal{E}_{\Delta^{d-1}} \subseteq \mathcal{H}_{\vec{1},1}$ by construction.

One way to establish that $\mathcal{E}_{\Delta^{d-1}}$ is furthermore itself an Euclidean ball, is using ``generalized cylindrical coordinates'' for the Euclidean unit ball $\mathcal{B}^d(\vec{0},1)$: 
Such coordinates use the fact that $\mathcal{B}^d (\vec{0},1)$ is equivalent to the union of a family of $(d-1)$-dimensional unit balls. 
More concretely: let $\vec{z} \in \mathbb{R}^d$ be an arbitrary unit vector and let $\zeta \in \mathbb{R}$ denote a parameter.
For each value of this parameter, we define the hyperplane $\tilde{\mathcal{H}}_{\vec{z}, \zeta} = \left\{\vec{s} \in \mathbb{R}^d: \; \langle \vec{z},\vec{s} \rangle = \zeta \right\}$ which in particular contains the vector $\zeta \vec{z}$ by construction. 
Furthermore, let $\tilde{\mathcal{B}}^{d-1} (\vec{z},\zeta)\subset \tilde{\mathcal{H}}_{z,\zeta}$ be the $(d-1)$-dimensional Euclidean ball with radius $\sqrt{1-\zeta^2}$ and center $\zeta \vec{z}$ that is contained in the hyperplane $\tilde{\mathcal{H}}_{\vec{z},\zeta}$. 
Clearly each element in such a union of sets is contained in the $d$-ball, and letting $\zeta$ range from $-1$ to $1$ covers the entire $d$-ball. 
In order to see this, decompose any  $\vec{s} \in \mathcal{B}^d (\vec{0},1)$ as $\vec{s} = \langle \vec{s},\vec{z} \rangle \vec{z} + \vec{z}^\perp$ such that $\langle \vec{z}^\perp,\vec{z} \rangle = 0$ and set $\zeta = \langle \vec{s},\vec{z}\rangle$. Pythagoras' theorem then assures $\| \vec{z}^\perp \|_2 \leq \sqrt{1 - \zeta^2}$ and consequently $\vec{s} \in \tilde{\mathcal{B}}^{d-1}(\vec{z},\zeta)$.

The structure of the particular problem at hand suggests to fix $\vec{z} = \frac{1}{\sqrt{d}} \vec{1}$. Indeed, such a particular choice of $\vec{z}$ assures equality of the hyperplane $\mathcal{H}_{\vec{1},1}$ which contains $\mathcal{E}_{\Delta^{d-1}}$ and
the hyperplane $\tilde{\mathcal{H}}_{\frac{1}{\sqrt{d}}\vec{1},\frac{1}{\sqrt{d} }}$,
we have just introduced.
Consequently, the ``cylindrical representation'' of the Euclidean unit ball assures that the intersection $\mathcal{E}_{\Delta^{d-1}} = \mathcal{B}_1 (0) \cap \mathcal{H}_{\vec{1},1}$ 
corresponds to the $(d-1)$-ball $\tilde{\mathcal{B}}^{d-1} (\frac{1}{\sqrt{d}}\vec{1},\frac{1}{\sqrt{d}})$ associated with the hyperplane $\tilde{\mathcal{H}}_{\frac{1}{\sqrt{d}}\vec{1},\frac{1}{\sqrt{d} }}$ and a parameter value $\zeta = \frac{1}{\sqrt{d}}$. 
By definition, this ball has center $\frac{1}{d}\vec{1}$ and radius $\sqrt{1- \zeta^2} = \sqrt{\frac{d-1}{d}}$ which completes the proof. 
\end{proof}

The next statement establishes that our choice of replacing the original feasible set $\Delta^{d-1}$ 
in the proof of Theorem 3
by the larger convex set $\mathcal{E}_{\Delta^{d-1}}$ is in a precise sense the tightest possible elliptic relaxation of the original optimization problem.

\begin{proposition}
The set $\mathcal{E}_{\Delta^{d-1}}$ is the unique minimal volume covering ellipsoid of the standard simplex $\Delta^{d-1}$. 
\end{proposition}

The proof exploits the following standard result about L\"owner-John ellipsoids that is originally due to John. However, here we make use of a slightly more general version presented in \cite{henk}.

\begin{theorem}[Theorem 2.1 in \cite{henk}] \label{thm:john}
Let $K \subset \mathbb{R}^d$ be a convex body and let $K$ be contained in the Euclidean unit ball $\mathcal{B}^d(0)$. Then the following statements are equivalent:
\begin{enumerate}
\item $\mathcal{B}^d (0)$ is the unique minimum volume ellipsoid containing $K$.
\item There exist contact points $\vec{u}_1,\ldots,\vec{u}_m$ lying both in the boundary of $K$ and $\mathcal{B}^d(0)$, and positive numbers $\lambda_1,\ldots,\lambda_m$, $m \geq d$, such that
\begin{align}
\sum_{i=1}^m \lambda_i \vec{u}_i = \vec{0}
\quad \textrm{and} \quad
\sum_{i=1}^m \lambda_i | \vec{u}_i \rangle\!\langle \vec{u}_i | = \mathbbm{1}.
\end{align}
\end{enumerate}
\end{theorem}

\begin{proof}
In Proposition \ref{prop:geometry1} we have established that the set $\mathcal{E}_{\Delta^{d-1}}$ corresponds to a $(d-1)$-ball with radius $\sqrt{\frac{d-1}{d}}$ 
and center $\frac{1}{d}\vec{1}$ that (like the standard simplex) is contained in the hyperplane $\mathcal{H}_{\vec{1},1}$.
A quick calculation reveals that all vertices of the standard simplex $\Delta^{d-1}$---which are just the standard basis vectors $\vec{e}_1,\ldots,\vec{e}_d$---have Euclidean distance $\sqrt{\frac{d-1}{d}}$ to the ball's center. Consequently they are contained in the boundary of the ball $\mathcal{E}_{\Delta^{d-1}}$ and we have found sufficiently many contact points for applying \autoref{thm:john}.
Since volume is translationally invariant we can furthermore shift the coordinate's origin into the point $\frac{1}{d} \vec{1}$ (which is the center of the ball $\mathcal{E}_{\Delta^{d-1}}$). This has the advantage that the affine space $\mathcal{H}_{\vec{1},1}$ containing both $\Delta^{d-1}$ and $\mathcal{E}_{\Delta^{d-1}}$ turns into $\mathcal{H}_{\vec{1},0}$  which is a linear subspace. 
Note that with respect to the (translated) standard basis, the orthogonal projection onto this subspace is given by
\begin{align*}
P = \mathbbm{1} - \frac{1}{d} | \vec{1} \rangle\!\langle \vec{1} |. \label{eq:projector}
\end{align*}
With respect to this new coordinate system, the $d$ contact points (vertices of the simplex) amount to $\vec{\tilde{e}}_i = \vec{e}_i - \frac{1}{d} \vec{1}$. 
Choosing unit weights $\lambda_i = 1$ for all $m = d$ contact points $\vec{u}_i = \vec{\tilde{e}}_i$ and calculating
\begin{align}
\sum_{i=1}^m \lambda_i \vec{u}_i
&= \sum_{i=1}^n \vec{\tilde{e}}_i
= \sum_{i=1}^n \left( \vec{e}_i - \frac{1}{d} \vec{1} \right) = \vec{0}
\end{align}
reveals that the first condition for \autoref{thm:john} is fulfilled.
A similar calculation reveals
\begin{align*}
\sum_{i=1}^m \lambda_i |\vec{u}_i \rangle\!\langle \vec{u}_i |
= \mathbbm{1} - \frac{1}{d} | \vec{1} \rangle\!\langle \vec{1} |.
\end{align*}
This, however equals just the projector $P$ onto the subspace $\mathcal{H}_{\vec{1},0}$ which contains the entire $(d-1)$-dimensional problem of interest. 
Restricted to its range, a projector corresponds to the identity which establishes the second condition for \autoref{thm:john}.
Since this statement is invariant under re-scaling, we can also apply it here, where the radius of the $(d-1)$-dimensional surrounding Euclidean ball is not one but $\sqrt{\frac{d-1}{d}}$. 
\end{proof}

\section{Conclusion}

In this work we have derived upper bounds on the average fidelity of any estimator with no restrictions on the dimension or the distribution being averaged over.  Furthermore, we have shown a sharp distinction in the optimization problems of maximizing average fidelity between measures supported only on pure states and those with full support.  In the former case, we have provided the exact optimal estimator, while in both cases we argued based on numerical evidence that the mean estimator is a good proxy for the optimal solution.  

Interestingly, we found that the analytical bound \eqref{eq:main_bound} (which is based on super-fidelity \cite{Miszczak2009Sub}) is strictly tighter than a corresponding one obtained using the well known, and often used, Fuchs-van de Graaf inequalities \cite{Fuchs1999Cryptograhic}. 

These results have obvious applications to practical Bayesian quantum tomography \cite{BlumeKohout2010Optimal}, since the bound can be computed \emph{online}---that is, it is only a property of the current distribution under consideration.  But we also expect our bound to be of interest in other theoretical work on tomography, where a benchmark is needed to make statements about absolute average performance of some candidate protocol.  

\begin{acknowledgments}
We thank Robin Blume-Kohout, Steve Flammia and Patrick Hayden for fruitful discussions.
CF was supported by National Science Foundation grant number PHY-1212445, the Canadian Government through the NSERC PDF program, the IARPA MQCO program, the ARC via EQuS project number CE11001013, and by the US Army Research Office grant numbers W911NF-14-1-0098 and W911NF-14-1-0103.  RK was  supported by scholarship funds from the State Graduate Funding Program of Baden-W\"urttemberg,
 the Excellence Initiative of the German Federal and State Governments (Grant ZUK 43),
the ARO under contracts, W911NF-14-1-0098 and W911NF-14-1-0133 (Quantum Characterization, Verification, and Validation),
the Freiburg Research Innovation Fund, and the DFG (GRO 4334-1-1). 
\newline
Furthermore, the authors would like to thank the \emph{Mathematisches Forschungsinstitut Oberwolfach} and the organizers of the Oberwolfach workshop \emph{New Horizons in Statistical Decision Theory} (September 2014, Workshop ID: 1437)
where work on this project commenced.
\end{acknowledgments}

\end{document}